\documentclass{article}
\usepackage[utf8]{inputenc}
\usepackage{a4wide}
\usepackage{xcolor}
\usepackage{amsmath}
\usepackage{lscape}
\usepackage[font=small,labelfont=bf]{caption}
\usepackage{graphicx}
\usepackage{subcaption}
\usepackage{comment}
\usepackage{tabularx}
\usepackage{multirow}
\usepackage{algorithmic}
\usepackage[ruled,vlined]{algorithm2e}
\usepackage{amsmath,amssymb,amsfonts,amsthm} 
{
\newtheorem{theorem}{Theorem}
\newtheorem{lemma}{Lemma}
\newtheorem{assumption}{Assumption}
\theoremstyle{remark}

\newtheorem {fact}{Fact}
}
\usepackage[utf8]{inputenc}
\usepackage{amssymb}
\usepackage{multicol}
\usepackage{geometry}
\usepackage{hyperref}
\usepackage{authblk}
\usepackage{ulem}
\usepackage{moreverb}
\usepackage{amsmath,bm}
\usepackage{graphicx}
\usepackage{siunitx}
\usepackage{subcaption}

\title{Connected and Automated Vehicle Platoon Formation Control via Differential Games} 

\author[1]{Hossein B. Jond}
\author[2]{Aykut Y{\i}ld{\i}z}
\affil[1]{Department of Computer Science, VSB-Technical University of Ostrava, Ostrava-Poruba, Czech Republic}
\affil[2]{Department of Electrical and Electronics Engineering, TED University, Ankara, Turkey}

\begin{document}

\maketitle

\begin{abstract}
 In this study, the connected and automated vehicles (CAVs) platooning problem is resolved under a differential game framework. Three information topologies are considered here. Firstly, Predecessor-following (PF) topology is utilized where the vehicles control the distance with respect to the merely nearest predecessor via a sensor link-based information flow. Secondly, Two-predecessor-following topology (TPF) is exploited where each vehicle controls the distance with respect to the two nearest predecessors. In this topology, the second predecessor is communicated via a Vehicle-to-vehicle (V2V) link. The individual trajectories of CAVs under the Nash equilibrium are derived in closed-form for these two information topologies. Finally, general information topology is examined and the differential game is formulated in this context. In all these options, Pontryagin's principle is employed to investigate the existence and uniqueness of the Nash equilibrium and obtain its corresponding trajectories. In the general topology, we suppose numerical computation of eigenvalues and eigenvectors. All these approaches represent promising and powerful analytical representations of the CAV platoons under the differential games. Simulation experiments have verified the efficiency of the proposed models and their solutions.
 
\textbf{Keywords:} Connected and automated vehicle (CAV); differential game; formation control; Nash equilibrium; Pontryagin's principle  
\end{abstract}
\maketitle

\section{Introduction}\label{sec1}

Connected and automated vehicles (CAVs) and driver-less cars are promising technologies to enhance our daily experience of traveling with vehicles on the roads. Especially, in the Internet of Things (IoT) era that the world is connected more than ever CAVs are no exception. The most common group behavior of CAVs is forming and maintaining platoons with short headways on the road \cite{WOO2021103442}. To this end, CAVs must communicate via an information exchange topology via onboard sensors or Vehicle-to-vehicle (V2V) communications. Platooning can be defined as the cooperation of vehicles that move together, by maintaining a certain formation \cite{bergenhem2012overview}. Such coordination is achieved by exchanging local communication information among the vehicles \cite{abualhoul2013platooning}. There are three types of platooning methods as designated in \cite{madeleine2012vehicle}. Leader-Follower approach in \cite{dunbar2011distributed}, 
Behavior-Based Approach in \cite{antonelli2010nsb}, and Virtual Structure approach in \cite{ren2004decentralized}. The behavior-based approach is associated with Artificial Potential Field (APF) in \cite{semsar2016cooperative}, flocking in \cite{hao2021flock}, and swarm intelligence in \cite{li2018swarm}. Platoons can be classified in terms of the type of vehicles as truck platoons \cite{tsugawa2013overview}, underwater vehicles \cite{stilwell2000platoons}, unmanned aerial vehicles
\cite{chen2017reachability}, and marine platoons \cite{liang2020adaptive}.

Vehicle platoons offer remarkable benefits as listed in \cite{wang2019survey}. First of all, road capacity can be boosted due to the decrease of gaps between vehicles. Secondly, fuel consumption and emissions of pollutants can be decreased owing to the elimination of dispensable change of speed and aerodynamic drag on the pursuing vehicles. Besides, driving security is potentially enhanced since the detection and actuation time are lessened compared to manual vehicles. Moreover, downstream traffic information can rapidly be transmitted upstream. Finally, passenger satisfaction can be enhanced since the system behavior is more reactive to changes in the traffic, and the shorter pursuing spaces can prevent cut-ins of other vehicles in the platoon. To sum up, vehicle platooning is a solution for traffic congestion and air pollution issues as stated in \cite{bergenhem2012overview}. 

Differential game \cite{hoogendoorn2009generic} and optimal control theory \cite{morbidi2013decentralized} have been widely applied to vehicle platooning problem. Optimal control theory studied in \cite{kirk2004optimal} considers a single cost function optimized by the agents. On the other hand, differential games in 
\cite{engwerda2005lq} extend this idea to multiple cost functions each minimized by its relevant agent. Differential game is ideally suited to model noncooperative dynamics where there is no central decision-maker. Thus, we can construct simpler systems since there is no need for such a central processor. This reduces the construction cost of the resulting system substantially. This ad hoc decentralized system notion is also utilized in 5G communication systems \cite{han2015green}.

Research on differential game-based control of networked multiagent systems has developed in three directions. These are Lyapunov stability analysis in \cite{zhong2017model} and numerical solution of Riccati Equations in \cite{gu2007differential} and numerical solution of Hamilton-Jacobi-Bellman (HJB) equations in \cite{haurie1994monitoring}. However, explicit solution of players' trajectories in differential game systems is studied only in a few research works and is left understudied in differential games-based networked control systems for vehicle platoons. Such explicit solutions enable fast simulations, reliable analysis, useful structural properties as well as a simple stability test. The main bottleneck of optimal control and differential games is the computational complexity of the underlying algorithms. Nevertheless, high complexity is not an issue for explicit expressions obtained in this study. Another fast computational scheme is the Model Predictive Control (MPC) in \cite{la2020supervised} which usually relies on linear approximations of the nonlinear problem over a small operating range. However, the proposed analytical approach is theoretically stronger than the MPC since does not rely on approximations.

In this study, our main contribution is to show that cooperative dynamics emerge from noncooperative decisions of vehicles. The problem that we examine is the one-dimensional point-to-point transportation of a number of vehicles that have coupled dynamics. The formulation of noncooperative games is based on Section 6.5.1 of \cite{bacsar1998dynamic}. We extend the PF topology investigated in \cite{yildiz2021vehicle} to the TPF and general topology in this study. This problem is crucial since it provides analytical insight into the platoon dynamics problem. We arrive at explicit Nash equilibrium derived in closed-form for the PF and TPF information topologies. We also obtain the solution of the differential game for general information topology that can be calculated numerically. 

The rest of the manuscript is organized as follows. In Section \ref{section:formulation}, the problem definitions, as well as main results and proofs, are highlighted. In Section \ref{section:simulation}, simulation results are demonstrated. Finally, Section \ref{section:conclusions} is dedicated to the conclusions and future works.

\section{Platoon formation control}\label{section:formulation}

This section presents the game theory approach to vehicle platoon formation control. In particular, the platoon is characterized based on the communication topology (CT) and inter-vehicle distances between each CAV and the preceding. A differential game theory formulation of the platoon formation control under PF and TPF topology are presented in the subsections \ref{section:PF} and \ref{section:TPF}, respectively. We have derived the corresponding Nash equilibrium individual state trajectories for the vehicles. Later, we generalized the differential game formulation of the platoon formation control to a general topology. However, we have not derived individual trajectories for the general topology since the posed differential game problem is analytically intricate due to its complex information structure (for example, see \cite{jond2021autonomous}).        

\subsection{Predecessor-following (PF) Topology}\label{section:PF}

Information exchange among the vehicles via a connected information topology is necessary to platoon formation control. Following \cite{9663029}, we distinguish between connection links
among vehicles by categorizing them into V2V connection links and
sensor connection links. A sensor link emphasis the information flow from the predecessor vehicle to the ego vehicle (see Figure \ref{fig:PF-topology}). An V2V connection link between a pair of vehicles can be one of rearward, forward, and bidirectional connection link types. 

Consider a platoon of $n$ CAVs under the PF topology as shown in Figure \ref{fig:PF-topology}. CAVs move along the center of a single lane all in the same traffic flow. Under this topology, each vehicle must follow its predecessor by adjusting its inter-vehicle distance $d_i$ by using the relative distance information that receives through the unidirectional sensor link. We assume that the actual leading vehicle of the platoon indexed by 1 follows a reference trajectory or a virtual leading vehicle indexed by 0. The reference ($x_0$) is either a manned/unmanned vehicle on the road that does not have an intention to establish V2V communications with the vehicles of the platoon or simply a reference trajectory generated by a driver-assistive technology such as Lane Following Assist (LFA). The actual vehicles of the convey are indexed as $1,\ldots,n$.

\begin{figure}[ht]
\centering
\includegraphics[width=1\textwidth]{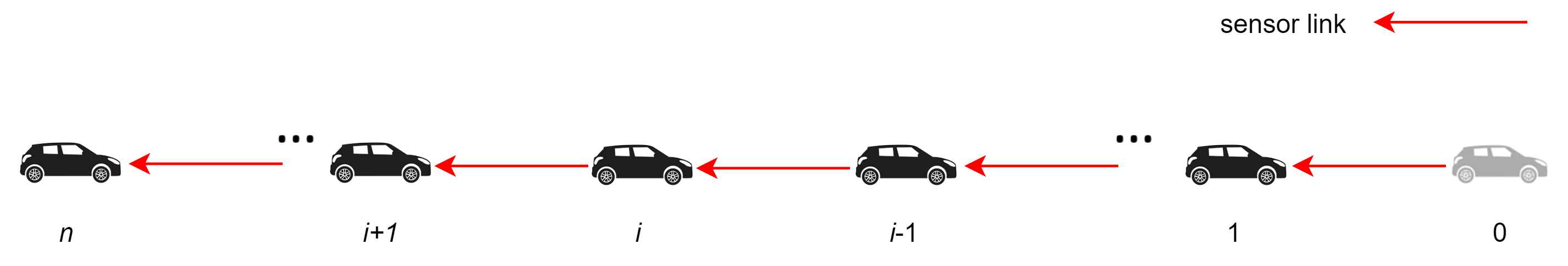}
\caption{A vehicle platoon with predecessor-following (PF) topology.} \label{fig:PF-topology}
\end{figure}

Let $x_i$ and $u_i$ denote the $i$th vehicle's longitudinal position and control input ($i=1,\ldots,n$), respectively. The inter-vehicles distancing policies longitudinal dynamics can be described as  
\begin{equation}\label{eq:pairwise-dyn}
    \dot x_i-\dot x_{i-1}=u_i.
\end{equation}
The platoon formation control cost function for vehicle $i$ ($i=1,\ldots,n$) is defined as: 
\begin{equation}\label{eq:costs}
    J_i(x_{i},x_{i-1},u_i)=\frac{1}{2}\int_0^{t_f}(\omega_i(x_{i}-x_{i-1}-d_i)^2+u_i^2)~dt,
\end{equation}
where $t_f$ is the finite time horizon and $\omega_i>0$ is a sensor link-related weighting parameter.  
We note that the PF topology is a collision-free platoon formation. Using the sensory information, each vehicle is aware of its relative distance to its predecessor and regulates its velocity to acquire and preserve the inter-vehicle distance then.  Since there is no V2V communication among the vehicles, this information exchange network is not subject to external hacks or faulty information sent by a vehicle. Therefore, a connected environment or either connected autonomous vehicles are not necessities for this model.      

\begin{assumption} (Collision-avoidance) \label{initial-position}
The initial positions hold $x_0(0)>x_1(0)>\ldots>x_n(0)$ and $d_i<0$ ($i=1,\ldots,n$) since the its direction is opposite to the traffic flow. \end{assumption} 

Under Assumption \ref{initial-position}, collision-avoidance is embedded in the platoon formation control given by (\ref{eq:pairwise-dyn}) and (\ref{eq:costs}).

\begin{assumption} (Reference trajectory)\label{leader-dynamics}
We assume $\dot x_0=0$, i.e., the reference has a constant speed. 
\end{assumption}

The dynamics and cost function in (\ref{eq:pairwise-dyn}) and (\ref{eq:costs}) demand the position information of vehicle and its predecessor. In the following, we transform these equations into a new dynamics and cost function that does not require the individual positions but the relative distance that can be obtained via the sensor links.

Let's define the relative distance $y_i=x_i-x_{i-1}$ for vehicle $i$ ($i=1,\ldots,n$) that can be directly measured via the sensor link. Therefore, the platoon formation control in (\ref{eq:pairwise-dyn}) and (\ref{eq:costs}) as a non-cooperative differential game can be redefined as the minimization of the following optimization:
\begin{align}\label{eq:PF-opt}
    &\quad\min_{u_i}~  J_i(y_{i},u_i)=\frac{1}{2}\int_0^{t_f}(\omega_i(y_{i}-d_i)^2+u_i^2)~dt,\\
    &{s.t.}\notag\\ 
    &\quad \dot y_i =u_i , \quad y_i(0)=x_i(0)-x_{i-1}(0), \quad i=1,\ldots,n \nonumber.
\end{align}

We present the closed-form solutions to individual (relative distances) trajectories corresponding to the Nash equilibrium as follows.

\begin{theorem}\label{theorem:PF}
 For a $n$-vehicle platoon defined as the non-cooperative differential game (\ref{eq:PF-opt}):  
\begin{enumerate}
    \item There is a unique Nash equilibrium of relative distances control actions $u_i$'s. 
    \item The Nash equilibrium's corresponding relative distances trajectories $y_i$'s are given by
    \begin{equation} \label{eq:pairwise}
    y_i(t)=\alpha_i(t)y_i(0)+(\alpha_i(t)-1)d_i,
\end{equation}
where
\begin{align}
    &\alpha_i(t)=\frac{\mathrm{cosh}(\sqrt{\omega_i}(t_f-t))}{\mathrm{cosh}(\sqrt{\omega_i}t_f)}.
\end{align}
\end{enumerate}
\end{theorem}
\begin{proof} 
Define the Hamiltonian
\begin{equation}\label{eq:non-coop-hamil}
    H_i(y_i,u_i,\lambda_i)=\frac{1}{2}(\omega_i(y_i-d_i)^2+u_i^2)+\lambda_i u_i, \quad i=1,\ldots,n,
\end{equation}
where $\lambda_i$ is the costate. According to the Pontryagin’s minimum principle, the necessary conditions for optimality are $\frac{\partial H_i}{\partial u_i}=0$ and $\dot \lambda_i=-\frac{\partial H_i}{\partial y_i}$. Applying the necessary conditions on (\ref{eq:non-coop-hamil}) yields:
\begin{align} \label{eq:neccesary-u}
    &u_i=-\lambda_i\\\label{eq:neccesary-lam},
    &\dot \lambda_i=-\omega_i y_i + \omega_i d_i,\quad ,
    \lambda_i(t_f)=0  
\end{align}
for $i=1,\ldots,n$.

Substituting (\ref{eq:neccesary-u}) in relative dynamics results in
\begin{equation}\label{eq:pairwise-lambda}
    \dot y_i= -\lambda_i, \quad y_i(0)=x_i(0)-x_{i-1}(0), \quad i=1,\ldots,n. 
\end{equation}

Let $\bm{y}=[y_1,\ldots,y_n]^\top\in \mathit{\mathbb{R}}^n$, $ \bm{\lambda}=[\lambda_1,\ldots,\lambda_n]^\top\in \mathit{\mathbb{R}}^n$, $ \bm{d}=[ d_1,\ldots,d_n]^\top\in \mathit{\mathbb{R}}^n$, $\bm{A}=\mathrm{diag}(\omega_1,\ldots,\omega_n)\in \mathit{\mathbb{R}}^n$, and $\bm{\omega}=\bm{A}\bm{d}$. Also, let ${\bf 0}$ and $\bm{I}$ denote the zero and identity matrix of appropriate dimension. The equations (\ref{eq:neccesary-u}), (\ref{eq:neccesary-lam}) and (\ref{eq:pairwise-lambda}) can be unified into the following differential equation
\begin{equation}\label{eq:state-equation}
    \begin{bmatrix}
\dot{\bm{y}} \\
\dot {\bm{\lambda}}
\end{bmatrix}=\begin{bmatrix}
{\bf 0} & -\bm{I} \\ -\bm{A} & {\bf 0}
\end{bmatrix}
\begin{bmatrix}
\bm{y} \\
\bm{\lambda}
\end{bmatrix}+
\begin{bmatrix}
{\bf 0}\\ \bm{\omega}
\end{bmatrix},
\end{equation}
with the initial condition vector $\bm{y}(0)$ where $y_i(0)=x_i(0)-x_{i-1}(0)$ ($i=1,\ldots,n$) and terminal condition vector $\bm{\lambda}(t_f)$ where $\lambda_i(t_f)=0$ ($i=1,\ldots,n$). 

If there is a unique Nash equilibrium solution, then equation (\ref{eq:state-equation}) has a solution for each $\bm{y}(0)$ and $\bm{\lambda}(t_f)$. Matrix analyses in Laplace Transform domain then shows that (\ref{eq:state-equation}) has a solution as follows
\begin{equation}\label{eq:state-sol}
    \begin{bmatrix}
\bm{y} \\
\bm{\lambda}
\end{bmatrix}={\bf\Phi}(t)
\begin{bmatrix}
\bm{y} (0)\\
\bm{\lambda}(0)
\end{bmatrix}+
{\bf \Psi}(t,0)\bm{\omega},
\end{equation}
where
\begin{align}
\bm{\Phi}(t)=&
\begin{bmatrix} \bm{\Phi}_{11}(t) & \bm{\Phi}_{12}(t)\\
\bm{\Phi}_{21}(t) &\bm{\Phi}_{22}(t)
\end{bmatrix}=\mathcal{L}^{-1}\{ 
    \begin{bmatrix} 
    s\bm{I} & \bm{I} \\
    \bm{A} & s\bm{I} 
    \end{bmatrix}^{-1}\}= \\ \nonumber
    &\mathcal{L}^{-1}\{ 
    \begin{bmatrix} 
    s(s^2\bm{I}-\bm{A})^{-1}   &    - (s^2\bm{I}-\bm{A})^{-1} \\
    -\bm{A}(s^2\bm{I}-\bm{A})^{-1}    &     s(s^2\bm{I}-\bm{A})^{-1}
    \end{bmatrix}\}=
\begin{bmatrix} \mathrm{cosh}(\bm{A}^{\frac{1}{2}}t) & -\mathrm{sinh}(\bm{A}^{\frac{1}{2}}t)\bm{A}^{-\frac{1}{2}}\\
-\bm{A}^{\frac{1}{2}}\mathrm{sinh}(\bm{A}^{\frac{1}{2}}t) & \mathrm{cosh}(\bm{A}^{\frac{1}{2}}t)
\end{bmatrix},    
\end{align}
and
\begin{equation}
    \bm{\Psi}(t,0)    =
    \begin{bmatrix}
    \bm{\Psi}_{1}(t,0)\\
    \bm{\Psi}_{2}(t,0)
    \end{bmatrix}
    =\int_0^{t}
    \begin{bmatrix}
    \bm{\Phi}_{12}(t-\tau)\\
    \bm{\Phi}_{22}(t-\tau)
    \end{bmatrix}d\tau=
    \begin{bmatrix}
    -(\bm{I}-\mathrm{cosh}(\bm{A}^{\frac{1}{2}}t))\bm{A}^{-1}\\
    -\mathrm{sinh}(\bm{A}^{\frac{1}{2}}t)\bm{A}^{-\frac{1}{2}}
    \end{bmatrix}.
\end{equation}

From (\ref{eq:state-sol}), $\bm{y}$ and $\bm{\lambda}$ are obtained as
\begin{align}
   \label{eq:PF-trajectory}&\bm{y}=\bm{\Phi}_{11}(t)\bm{y}(0)+\bm{\Phi}_{12}(t)\bm{\lambda}(0)+\bm{\Psi}_1(t,0)\bm{\omega},\\
   &\bm{\lambda}=\bm{\Phi}_{21}(t)\bm{y}(0)+\bm{\Phi}_{22}(t)\bm{\lambda}(0)+\bm{\Psi}_2(t,0)\bm{\omega}.
\end{align}
Applying the terminal condition $\bm{\lambda}(t_f)={\bf 0}$, we obtain vector $\bm{\lambda}(0)$ as
\begin{equation}\label{eq:initial-lambda-PF}
    \bm{\lambda}(0)=\bm{\Phi}_{22}^{-1}(t_f) \left[-\bm{\Phi}_{21}(t_f)\bm{y}(0) -\bm{\Psi}_2(t_f,0)\bm{\omega}\right].
\end{equation}
Notice that the existence of $\bm{\Phi}_{22}^{-1}(t_f)$ (which also means the existence of a unique Nash equilibrium solution) and $\bm{\Psi}_2(t_f,0)$ for all $t_f>0$ can be realized from the definition of the hyperbolic functions. Since $\bm{A}$ is diagonal, therefore, $\mathrm{sinh}(\bm{A}^{\frac{1}{2}}t_f)$ and $\mathrm{cosh}(\bm{A}^{\frac{1}{2}}t_f)$ are diagonal matrices with the diagonal elements of $\mathrm{sinh}(\sqrt{\omega_i}t)>0$ and $\mathrm{cosh}(\sqrt{\omega_i}t)>0$ for $i=1,\ldots,n$, respectively as $\omega_i>0$.      

Finally, substituting (\ref{eq:initial-lambda-PF}) into (\ref{eq:PF-trajectory}) and after simplifications, the relative distancing policies' trajectories $y_i$'s are obtained as in (\ref{eq:pairwise}).
\end{proof}

The unique Nash equilibrium control actions $u_i$'s and the individual vehicle trajectories $x_i$'s can be calculated from (\ref{eq:pairwise-dyn}) and the relative distances given by (\ref{eq:pairwise}).

\subsection{Two-predecessor-following (TPF) Topology}\label{section:TPF}

Sometimes during the platoon formation control the sensor link can be unavailable due to, for example, hardware problems in the sensors. In such a situation, collision with the preceding or following vehicle is an immediate danger since the formation control cannot be maintained. Therefore, a more reliable than PF but still relatively simple vehicle platoon topology can be Two-predecessor-following (TPF) topology as shown in Figure \ref{fig:TPF}. Under this topology, if, for example,  sensor link from the preceding vehicle is not available to vehicle $i$, a rearward V2V link from the second preceding vehicle is still available to maintain the platoon formation control.  

\begin{figure}[ht]
\centering
\includegraphics[width=1\textwidth]{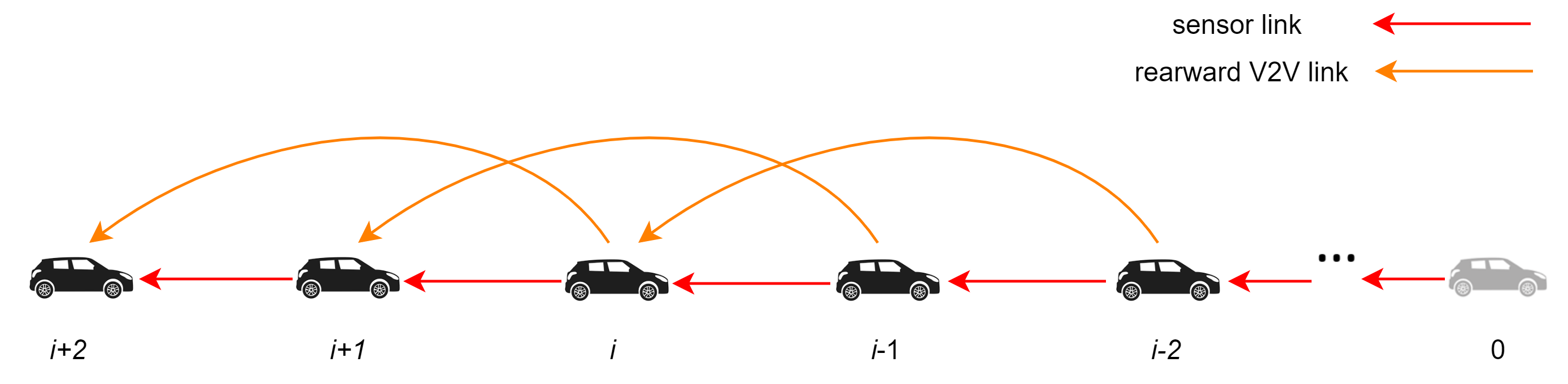}
\caption{Two-predecessor-following (TPF) topology in the connected environment.} \label{fig:TPF}
\end{figure}

 The platoon formation control cost function for vehicle $i$ ($i=1,\ldots,n$) under the TPF topology is defined: 
 \begin{equation}\label{eq:TPF-cost}
    J_i(x_{i},x_{i-1},x_{i-2}, u_i)=\frac{1}{2}\int_0^{t_f}(\omega_i(x_{i}-x_{i-1}-d_i)^2+\Tilde{\omega}_i\left[ x_{i}-x_{i-2}-(d_i+d_{i-1})\right]^2+u_i^2)~dt
\end{equation}
where $\Tilde{\omega}_i>0$ a V2V communication link-related weighting parameter.

 Using the relative distances $y_i$, the platoon formation control under the TPF topology as a non-cooperative differential game can be redefined as the minimization of the following optimization:

\begin{align}\label{eq:opt-TPF}
    &\quad \min_{u_i}~  J_i(y_i,y_{i-1},u_i)=\frac{1}{2}\int_0^{t_f}(\omega_{i}(y_{i}-d_{i})^2+ \Tilde{\omega}_{i}\left[(y_{i}- d_i)+(y_{i-1}- d_{i-1})\right]^2+u_i^2)~dt\\
    &{s.t.}\notag\\
    &\quad \dot y_i =u_i , \quad, y_i(0)=x_i(0)-x_{i-1}(0), \quad i=1,\ldots,n \nonumber
\end{align}

Before we present the closed-form solutions to individual (relative distances) trajectories corresponding to the Nash equilibrium, the following facts and lemma are given first. 
\begin{fact} \label{fact:det}
The determinant of an upper (or lower) triangular or diagonal matrix is the product of the (main) diagonal entries \cite{ANDRILLI2010143}.
\end{fact}
\begin{fact} \label{fact:inverse}
Let $\bm{B}\in \mathit{\mathbb{R}}^n$ be a nonsingular matrix. Then, $\bm{B}^{-1}$ is upper (lower) triangular if and only if $\bm{B}$ is upper (lower) triangular (Fact 3.20.5. in \cite{bernstein11}).
\end{fact}

\begin{lemma}\label{lemma:factor} Matrix 
\begin{align}
\bm{A}=&\begin{bmatrix}
    \omega_{1} & 0 & 0 & 0 &\cdots & 0 \\
    0 & \omega_{2} & 0 & 0 &\cdots & 0  \\
    0 & \Tilde{\omega}_3 & \Tilde{\omega}_3+\omega_{3}& 0 & \cdots & 0  \\
    \vdots & \ddots & &\ddots & &\vdots  \\
    0 & 0& \cdots &0 & \Tilde{\omega}_n  &  \Tilde{\omega}_n+\omega_{n} 
    \end{bmatrix}
\end{align}
can be factorized as:
\begin{equation}\label{eq:factor}
   \bm{A}=\bm{V}\bm{D}\bm{V}^{-1}
\end{equation}
where 
\begin{align}
    &\bm{D}=\mathrm{diag}[\cdots,\delta_i,\cdots],\label{eq:diag-D}\\ 
    &\bm{V}=[\cdots,\bm{v}_i,\cdots]\label{eq:modal-V}\\
    &\bm{V}^{-1}=[\cdots,\bm{\Tilde{v}}_i,\cdots]\label{eq:inv-v}\\
  &\delta_i=
\left\{
    	\begin{array}{ll}
		\omega_i  \quad &\quad \mbox{, $i=1,2$ } \\
		&\\
		 \Tilde{\omega}_i+\omega_i \quad & \quad\mbox{, $i=3,\cdots,n$ } \\
	\end{array}
	\right.\label{eq:eignvalue-delta}\\ 
	&\bm{v}_1=[\eta_1^1,0,\cdots,0]^\top \label{eq:v1}\\
	&\bm{v}_i=[0,\cdots,0,\eta_i^j,-\frac{\Tilde{\omega}_j  \eta_i^{j-1}}{\Tilde{\omega}_j+\omega_j-\delta_i},\cdots,-\frac{\Tilde{\omega}_n \eta_i^{n-1}}{\Tilde{\omega}_n+\omega_n-\delta_i}]^\top, \quad j=i,\cdots,n,\quad i=2,\cdots,n-1 \label{eq:vi}\\
	&\bm{v}_n=[0,\cdots,0,\eta_n^n]^\top \label{eq:vn}\\
	&\bm{\Tilde{v}}_1=[\zeta_1^1,0,\cdots,0]^\top\label{eq:inv-v1}\\
	&\bm{\Tilde{v}}_i=[0,\cdots,0,\zeta_i^j,\frac{\Tilde{\omega}_j-\eta_i^j\zeta_i^i\delta_i}{\eta_j^j\delta_j},-\frac{\sum_{k=i}^{j-1}\eta_k^j\zeta_i^k \delta_k}{\eta_j^j\delta_j},\cdots]^\top, \quad j=i,\quad i=2,\cdots,n-1\label{eq:inv-vi}\\
	&\bm{\Tilde{v}}_n=[0,\cdots,0,\zeta_n^n]^\top\label{eq:inv-n}\\
	&\zeta_i^i=\frac{1}{\eta_i^i}, \quad i=1,\cdots,n\label{eq:inv-element}
\end{align}

\begin{proof}
The eigenvalues of $\bm{A}$ are the roots of its characteristic polynomial $\rho(\delta)$:
\begin{equation}\label{eq:eigenvalue-def}
\begin{aligned}
 \rho(\delta)=& \det(\bm{A}-\delta \bm{I}) 
\end{aligned}
\end{equation}
where $\delta$ is an unknown variable representing the unknown eigenvalues. According to Fact \ref{fact:det}, (\ref{eq:eigenvalue-def}) simplifies as 
\begin{equation}\label{eq:eigenvalue-pol}
\begin{aligned}
 \rho(\delta)=& \prod_{i=1}^{2}(\omega_i-\delta)\prod_{i=3}^{n}\left[(\Tilde{\omega}_i+\omega_i)-\delta)\right]
\end{aligned}
\end{equation}
Solving (\ref{eq:eigenvalue-pol}) for its roots, the eigenvalues of $\bm{A}$ are obtained as (\ref{eq:eignvalue-delta}). All of its eigenvalues are real and positive ($\delta_i>0$ for $i=1,\ldots,n$), therefore, obviously, $\bm{A}$ is positive definite.

From (\ref{eq:eigenvalue-pol}), it can be easily seen that the algebraic and geometric multiplicities of eigenvalue $\lambda_i$ are both equal to 1 for all $i=1,\ldots,n$. Therefore, matrix $\bm{A}\in \mathit{\mathbb{R}}^n$ has $n$ linearly independent eigenvectors and can be factorized as in (\ref{eq:factor}) where the modal matrix $\bm{V}=[\bm{v}_1,\cdots,\bm{v}_n]\in \mathit{\mathbb{R}}^n$ and diagonal matrix $\bm{D}=\mathrm{diag}(\lambda_1,\cdots,\lambda_n)\in \mathit{\mathbb{R}}^n$ are constituted from the eigenvectors and eigenvalues of $\bm{A}$, respectively.

Now, we show that vector $\bm{v}_i$ given by (\ref{eq:vi}) and (\ref{eq:vn}) are the right eigenvectors associated with eigenvalue $\delta_i$ ($i=2,\ldots,n$). The right eigenvector satisfies:
\begin{equation}\label{eq:eign-vec}
(\bm{A}-\delta_i \bm{I})\bm{v}_i=\bm{0}
\end{equation}
The $j$th row of $\bm{A}-\delta_i \bm{I}$ has two non-zero elements, $\Tilde{\omega}_j$ and $(\Tilde{\omega}_j+\omega_{j})-\delta_i$ at the indices $i-1$ and $i$, respectively. Then, the right side of (\ref{eq:eign-vec}) reduces to $\left[(\Tilde{\omega}_j+\omega_{j})-\delta_i\right]\eta_i^i$. Inside the bracket always equals to zero. Similarly, can be verified that $\bm{v}_1$ given by (\ref{eq:v1}) satisfies (\ref{eq:eign-vec}). Therefore, $\bm{v}_i$ satisfies (\ref{eq:eign-vec}) for all $i=1,\ldots,n$.    
Matrix $\Tilde{\bm{V}}$ from (\ref{eq:factor}) and using Fact \ref{fact:inverse} can be determined in form of (\ref{eq:inv-v}) and (\ref{eq:inv-v1})-(\ref{eq:inv-element}). 
\end{proof}
\end{lemma}

Note that based on Lemma \ref{lemma:factor}, we have:
\begin{align}\label{eq:factorsA}
   &\bm{A}^{\frac{1}{2}}=\bm{V}\bm{D}^{\frac{1}{2}}\bm{V}^{-1}, \quad \mathrm{sinh}(\bm{A}^{\frac{1}{2}})=\bm{V}\mathrm{sinh}(\bm{D}^{\frac{1}{2}})\bm{V}^{-1},\nonumber\\  &\mathrm{cosh}(\bm{A}^{\frac{1}{2}})=\bm{V}\mathrm{cosh}(\bm{D}^{\frac{1}{2}})\bm{V}^{-1},\quad (\mathrm{cosh}(\bm{A}^{\frac{1}{2}}))^{-1}=\bm{V}^{-1}\mathrm{cosh}(\bm{D}^{-\frac{1}{2}})\bm{V}
\end{align}

The closed-form of relative trajectories $y_i$'s corresponding to the Nash equilibrium are now presented as follows.

\begin{theorem}\label{theorem:TPF}
 For a $n$-vehicle platoon under the TPF topology defined as the non-cooperative differential game (\ref{eq:opt-TPF}):  
\begin{enumerate}
    \item There is a unique Nash equilibrium of relative distances control actions. 
    \item The Nash equilibrium's corresponding relative distances trajectories are given by
\begin{align}\label{eq:trajectories-TPF}
    y_i(t)=\left\{
    	\begin{array}{ll}
		\alpha_i^i(t)y_i(0)+(\alpha_i^i(t)-1)d_i \quad &\quad \mbox{$i=1,2$ } \\
		&\\
		 \sum_{k=2}^i\alpha_k^i(t)y_k(0)+\sum_{k=2}^{i-1}\alpha_k^i(t)d_k+(\alpha_i^i(t)-1)d_i \quad & \quad\mbox{$i=3,\cdots,n$ } \\
	\end{array}
	\right.
\end{align}
where
\begin{align}
    &    \alpha_i^i(t)=
\left\{
    	\begin{array}{ll}
		\frac{\mathrm{cosh}(\sqrt{\omega_i}(t_f-t))}{\mathrm{cosh}(\sqrt{\omega_i}t_f)}  \quad &\quad \mbox{$i=1,2$ } \\
		&\\
		 \frac{\mathrm{cosh}(\sqrt{\Tilde{\omega}_i+\omega_i}(t_f-t))}{\mathrm{cosh}(\sqrt{\Tilde{\omega}_i+\omega_i}t_f)} \quad & \quad\mbox{$i=3,\cdots,n$ } \\
	\end{array}
	\right.\label{eq:alpha_i^i}\\
    &\alpha_i^j(t)=\sum_{k=i}^j \eta_k^j\zeta_i^k \alpha_k^k(t) \quad  \quad, i\neq j, i=2,\cdots,n, j=3,\cdots,n\label{eq:alpha_i^j}
\end{align}
\end{enumerate}
\end{theorem}

\begin{proof} 
Applying the necessary conditions for optimality on the Hamiltonians yields:
\begin{align} \label{eq:neccesary-TPF}
    &u_i=-\lambda_i\\\label{eq:neccesary-lam2}
        &\dot \lambda_i=-\omega_{i}(y_{i}-d_{i})-\Tilde{\omega}_i(y_i-d_i+y_{i-1}- d_{i-1}),\quad \lambda_i(t_f)=0 \\
\mbox{and}\nonumber\\
    &\dot y_i= -\lambda_i,\quad y_i(0)=x_i(0)-x_{i-1}(0) 
\end{align}
for $i=1,\ldots,n$.

Applying the terminal condition $\bm{\lambda}(t_f)={\bf 0}$ and using (\ref{eq:factorsA}), vector $\bm{\lambda}(0)$ after simplification is obtained as
\begin{align}\label{eq:lambda0}
    \bm{\lambda}(0)=&(\mathrm{cosh}(\bm{A}^{\frac{1}{2}}t_f))^{-1} \left[\bm{A}^{\frac{1}{2}}\mathrm{sinh}(\bm{A}^{\frac{1}{2}}t_f)\bm{y}(0) +\mathrm{sinh}(\bm{A}^{\frac{1}{2}}t_f)\bm{A}^{-\frac{1}{2}}\bm{\omega}\right] \nonumber\\
    =& \bm{V}(\mathrm{cosh}(\bm{D}^{\frac{1}{2}}t_f))^{-1}\bm{V}^{-1} \left[\bm{V}\bm{D}^{\frac{1}{2}}\mathrm{sinh}(\bm{D}^{\frac{1}{2}}t_f)\bm{V}^{-1}\bm{y}(0) +\bm{V}\bm{D}^{\frac{1}{2}}\mathrm{sinh}(\bm{D}^{\frac{1}{2}}t_f)\bm{V}^{-1}\bm{d}\right]\nonumber\\
    =& \bm{V}(\mathrm{cosh}(\bm{D}^{\frac{1}{2}}t_f))^{-1}\bm{D}^{\frac{1}{2}}\mathrm{sinh}(\bm{D}^{\frac{1}{2}}t_f)\bm{V}^{-1}\left[\bm{y}(0) +\bm{d}\right]
\end{align}

Substituting $\bm{\lambda}(0)$ from (\ref{eq:lambda0}), $\bm{y}$ is simplified as
\begin{align}
   \label{eq:whole-trajectories}
      \bm{y}=&\mathrm{cosh}(\bm{A}^{\frac{1}{2}}t)\bm{y}(0)-\mathrm{sinh}(\bm{A}^{\frac{1}{2}}t)\bm{A}^{-\frac{1}{2}}\bm{\lambda}(0)-(\bm{I}-\mathrm{cosh}(\bm{A}^{\frac{1}{2}}t))\bm{d}\nonumber\\
    =&\bm{V}\mathrm{cosh}(\bm{D}^{\frac{1}{2}}t)\bm{V}^{-1}\bm{y}(0)-\bm{V}\mathrm{sinh}(\bm{D}^{\frac{1}{2}}t)\bm{D}^{-\frac{1}{2}}\bm{V}^{-1}\{\ \bm{V}(\mathrm{cosh}(\bm{D}^{\frac{1}{2}}t_f))^{-1}\bm{D}^{\frac{1}{2}}\mathrm{sinh}(\bm{D}^{\frac{1}{2}}t_f)\bm{V}^{-1}\left[\bm{y}(0)+\bm{d}\right]\}\nonumber\\+&\bm{V}\mathrm{cosh}(\bm{D}^{\frac{1}{2}}t)\bm{V}^{-1}\bm{d}-\bm{d} \nonumber\\
        =&\left[ \bm{V}\mathrm{cosh}(\bm{D}^{\frac{1}{2}}t)\bm{V}^{-1}-\bm{V}(\mathrm{cosh}(\bm{D}^{\frac{1}{2}}t_f))^{-1}\mathrm{sinh}(\bm{D}^{\frac{1}{2}}t)\mathrm{sinh}(\bm{D}^{\frac{1}{2}}t_f)\bm{V}^{-1}\right]\bm{y}(0)+\nonumber\\
    &\left[\bm{V}\mathrm{cosh}(\bm{D}^{\frac{1}{2}}t)\bm{V}^{-1}-\bm{V}(\mathrm{cosh}(\bm{D}^{\frac{1}{2}}t_f))^{-1}\mathrm{sinh}(\bm{D}^{\frac{1}{2}}t)\mathrm{sinh}(\bm{D}^{\frac{1}{2}}t_f)\bm{V}^{-1} -\bm{I}\right]\bm{d}
\end{align}

Let
\begin{align}
   \bm{P}=&\bm{V}\mathrm{cosh}(\bm{D}^{\frac{1}{2}}t)\bm{V}^{-1}-\bm{V}(\mathrm{cosh}(\bm{D}^{\frac{1}{2}}t_f))^{-1}\mathrm{sinh}(\bm{D}^{\frac{1}{2}}t)\mathrm{sinh}(\bm{D}^{\frac{1}{2}}t_f)\bm{V}^{-1}\nonumber\\
   =&\bm{V}\left[\mathrm{cosh}(\bm{D}^{\frac{1}{2}}t)-(\mathrm{cosh}(\bm{D}^{\frac{1}{2}}t_f))^{-1}\mathrm{sinh}(\bm{D}^{\frac{1}{2}}t)\mathrm{sinh}(\bm{D}^{\frac{1}{2}}t_f)\right]\bm{V}^{-1}\nonumber\\
    =&\bm{V}\left[\mathrm{cosh}(\bm{D}^{\frac{1}{2}}(t_f-t))(\mathrm{cosh}(\bm{D}^{\frac{1}{2}}t_f))^{-1}\right]\bm{V}^{-1}\nonumber\\
    =&\bm{V}\bm{\Delta}\bm{V}^{-1}
\end{align}
Here, $\bm{\Delta}=\mathrm{diag}[\cdots,\alpha_i^i,\cdots]$ where $\alpha_i^i$ are given by (\ref{eq:alpha_i^i}) and matrix $\bm{P}$ has the following form:
\begin{align}
&\bm{P}=\begin{bmatrix}
    \alpha_1^1 & 0 & 0 & 0 &\cdots & 0 \\
    0 & \alpha_2^2 & 0 & 0 &\cdots & 0  \\
    0 & \alpha_2^3 & \alpha_3^3 & 0 & \cdots & 0  \\
    \vdots & \vdots & &\ddots & &\vdots  \\
    0 & \alpha_2^n& \alpha_3^n&  \cdots & \alpha_{n-1}^n  &  \alpha_n^n  
    \end{bmatrix}
\end{align}
where its non-diagonal non-zero elements $\alpha_i^j$ are given by (\ref{eq:alpha_i^j}).

Now, we get vector $\bm{y}$ in terms of $\bm{P}$ as  
\begin{align}
   \label{eq:matrix-trajectories}
      \bm{y}=&\bm{P}\bm{y}(0)+(\bm{P} -\bm{I})\bm{d}
\end{align}
According to the structure of matrix $\bm{P}$, the individual trajectories $y_i$'s from (\ref{eq:matrix-trajectories}) are derived as (\ref{eq:trajectories-TPF}).
\end{proof}

\subsection{General Topology}

In this study, we assume that the vehicles are connected via the rearward V2V and sensor link so the information flow is only backward in the opposite direction of the traffic flow. In the following, we generalize the proposed differential game approach to platoon formation control under the PF and TPF topology to general topology.     

The platoon formation control cost function for vehicle $i$ ($i=1,\ldots,n$) under general topology is defined: 
\begin{equation}\label{eq:costs2}
    J_i(x_{1},\cdots,x_{n},u_i)=\frac{1}{2}\int_0^{t_f}(\sum_{j\in \mathcal{N}_i}\omega_{ij}(x_{i}-x_{j}-\sum_{k=j+1}^{i} d_k)^2+u_i^2)~dt
\end{equation}
where $\omega_{ij}>0$ and $\mathcal{N}_i$ is the set of vehicles in the platoon directly connected to the vehicle $i$ either via a sensor link or (rearward) V2V connection link. 

The platoon formation control under general topology as a non-cooperative differential game can be redefined as the minimization of the following optimization:
\begin{align}\label{eq:opt-gen-top}
    &\min_{u_i}~  J_i(y_{j\in \mathcal{N}_i},u_i)=\frac{1}{2}\int_0^{t_f}(\sum_{j\in \mathcal{N}_i}\omega_{ij}\left[\sum_{k=j+1}^{i} (y_k-d_k)\right]^2+u_i^2)~dt\\
    &\mbox{s.t.}\notag\\ 
    &\quad \dot y_i =u_i , \quad, y_i(0)=x_i(0)-x_{i-1}(0), \quad i=1,\ldots,n \nonumber
\end{align}

The information structure of the proposed differential games represented by matrix $\bm{A}$ under the PF and TPF topology are rather straightforward. Unlike in the cases of the aforementioned topology, due to its complex information structure, finding closed-form individual trajectories for platoon formation control under general topology as a non-cooperative differential game (\ref{eq:opt-gen-top}) is highly analytically intricate. The following theorem presents some results regarding solving (\ref{eq:opt-gen-top}).

\begin{theorem}\label{theorem:general-top}
 For a $n$-vehicle platoon under the general topology defined as the non-cooperative differential game (\ref{eq:opt-gen-top}):
 \begin{enumerate}
    \item There is a unique Nash equilibrium of relative distances control actions $u_i$'s. 
    \item The whole platoon's relative distances' trajectories corresponding to the Nash equilibrium are obtained as in (\ref{eq:PF-trajectory}).
\end{enumerate}
\end{theorem}

\begin{proof}
Applying the necessary conditions for optimality on the Hamiltonian yields:

\begin{align} \label{eq:lambda-gen-top}
    &\dot \lambda_i=-\sum_{j\in \mathcal{N}_i}\omega_{ij}\sum_{k=j+1}^{i} (y_k-d_k),\quad \lambda_i(t_f)=0 
\end{align}
for $i=1,\ldots,n$.

From (\ref{eq:lambda-gen-top}), the platoon information matrix $\bm{A}\in \mathit{\mathbb{R}}^n$ under general topology is defined as:
\begin{align}
    \bm{A}=&\begin{bmatrix}
    \sum_{j\in N_1}w_{1j} & 0 & \cdots & 0 \\
    \sum_{j\in N_2,j<1}w_{2j} & \sum_{j\in N_2}w_{2j} & 0 & 0  \\
    \vdots & \vdots &\ddots & \vdots  \\
    \sum_{j\in N_n, j<1} w_{nj} & \sum_{j\in N_n, j<2} w_{nj} & \cdots  & \sum_{j\in N_n} w_{nj} 
    \end{bmatrix}
\end{align}
 
 Following the solution procedure in proof of Theorem \ref{theorem:PF}, the differential equation (\ref{eq:state-equation}) can be formed where its solution is given by (\ref{eq:state-sol}). From this equation, the whole platoon's relative distances' trajectories corresponding to the Nash equilibrium are obtained as in (\ref{eq:PF-trajectory}). It can be seen that the calculation of the trajectories requires substituting vector $\bm{\lambda}(0)$ from (\ref{eq:initial-lambda-PF}). This vector exists if and only if matrices $\bm{\Phi}_{22}^{-1}(t_f)$ and $\bm{A}^{\pm\frac{1}{2}}$ exist. 
 
 The eigenvalues of $\bm{A}$ are obtained from its characteristic polynomial (\ref{eq:eigenvalue-def}). This equation simplifies as 
\begin{equation}\label{eq:eigenvalue-pol2}
\begin{aligned}
 \rho(\delta)=& \prod_{i=1}^{n}(\sum_{j\in N_i} w_{ij}-\delta)
\end{aligned}
\end{equation}
From (\ref{eq:eigenvalue-pol2}), the eigenvalues of $\bm{A}$ are $\delta_i=\sum_{j\in N_i} w_{ij}$ where all are real and positive ($\delta_i>0$ for $i=1,\ldots,n$) and $\bm{A}$ is positive definite. Additionaly, we can see that the algebraic and geometric multiplicities of eigenvalue $\lambda_i$ are both equal to 1 for all $i=1,\ldots,n$. Thus, there exist a diagonal matrix $\bm{D}$ and a modal matrix $\bm{V}$ in form of (\ref{eq:diag-D}) and (\ref{eq:modal-V}) and are constituted from the eigenvectors and eigenvalues of $\bm{A}$, respectively, and  $\bm{A}$ is factorized as in (\ref{eq:factor}). Consequently, matrices $\bm{\Phi}_{22}^{-1}(t_f)$ and $\bm{A}^{\pm\frac{1}{2}}$ exist and can be calculated from (\ref{eq:factorsA}).
\end{proof}

\section{Simulation Results}\label{section:simulation}
In this section, we will conduct a number of simulation experiments on vehicle platoon formation control under the PF and TPF topology to justify the effectiveness of the proposed models and verify the closed-form solutions. In addition, we carry simulations on the general topology to observe the behavior of the platoon formation control under the proposed differential game scheme. For this case, the closed-form solutions are not available, and thus the whole group of CAVs relative distances trajectories is calculated from (\ref{eq:PF-trajectory}).

The size of the simulation model is $n=5$ that includes 5 actual vehicles in the platoon and one reference vehicle/trajectory. The finite horizon time when the underlying topology are PF and TPF is considered to be $t_f=10$. Under these topology and applying the closed-form solutions in Theorem \ref{theorem:PF} and \ref{theorem:TPF} with the parameter values given in Table \ref{tab:PF} and \ref{tab:TPF}, the simulation results are illustrated under four different scenarios. The CAVs initial positions, their spacing policies, and weighting parameter values given in Tables \ref{tab:PF}-\ref{tab:general} are generated randomly.

\textbf{Scenario 1.} The weighting parameter values given on the left section of Table \ref{tab:PF} are generated randomly in range [0.1,1]. The time histories of the sum of the relative distances trajectories with the CAVs spacing policies are shown in Figure \ref{fig:scenario1}(\subref{fig:scenario1-y}). We can see that all the time histories approach zero that means each relative distance satisfies $\lim_{t\to\infty} y_i(t)+d_i=0$ (note that $d_i<0$). From Figure \ref{fig:scenario1}(\subref{fig:scenario1-u}), we observe that the corresponding control input of the formation control signal converges to zero at the end of the terminal time. In this scenario, CAVs 1 and 4 converge faster while CAVs 2 and 5 converge slower. 

\begin{table}[ht]\caption{Simulation model's parameter values under the PF topology}
\centering
 \begin{tabular}{|c| c c c | c c c |} 
 \hline
 & &  Scenario 1 & &  & Scenario 2 & \\
 \hline
 $i$& $\omega_i$ &  $d_i$ & $x_i(0)$ & $\omega_i$ &  $d_i$ & $x_i(0)$ \\ [0.5ex] 
 \hline
 0 &   -    &   -   & 5.0937 &   -    &   - & 4.3064 \\ 
 1 & 0.6443 & -0.1 & 4.6469 & 0.8258 & -0.2 & 3.6586\\ 
 2 & 0.3786 & -0.2 & 3.8786 & 0.5383 & -0.2 & 3.2456\\
 3 & 0.8116 & -0.2 & 2.4340 & 0.9961 & -0.1 & 0.8450\\
 4 & 0.5328 & -0.3 & 2.1793 & 0.0782 & -0.3 & 0.2151\\
 5 & 0.3507 & -0.3 & 0.4056 & 0.4427 & -0.1 & 0.0770\\ [1ex] 
 \hline
 \end{tabular}\label{tab:PF}
\end{table}

\begin{figure}
     \centering
     \begin{subfigure}[b]{0.49\textwidth}
         \centering
         \includegraphics[width=\textwidth]{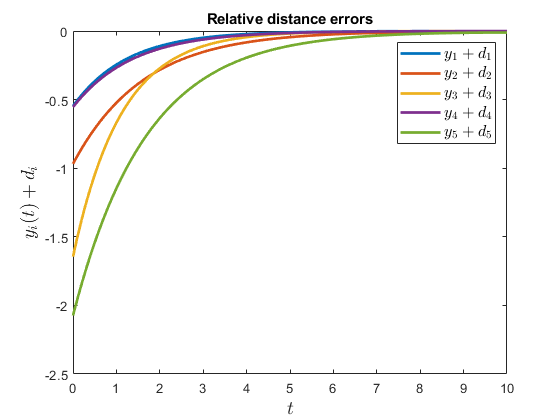}
         \caption{}
         \label{fig:scenario1-y}
     \end{subfigure}%
     \hfill
     \begin{subfigure}[b]{0.49\textwidth}
         \centering
         \includegraphics[width=\textwidth]{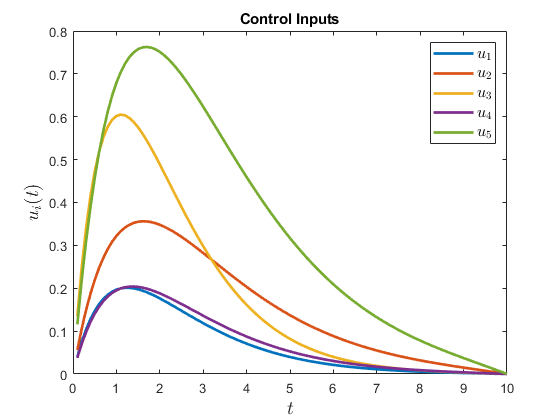}
         \caption{}
         \label{fig:scenario1-u}
     \end{subfigure}%
     \hfill
        \caption{Time histories of (a) relative distance errors and (b) control inputs under the PF topology. The parameter values are set according to Scenario 1 column of Table \ref{tab:PF}.}
        \label{fig:scenario1}
\end{figure}

\textbf{Scenario 2.} Another experiment with the platoon formation control model's parameter values given on the right section of Table \ref{tab:PF} is carried out. For this experiment, the range for generating the weighting parameter values is [0,1]. Two of these, $\omega_3$ and $\omega_4$, are relatively extreme, the former is big and the latter is small. From Figure \ref{fig:scenario2}(\subref{fig:scenario2-y}), we see that CAVs 3 and 4 relative distance trajectories converge faster and slower, respectively. Similar observation can be made on their control inputs from Figure \ref{fig:scenario2}(\subref{fig:scenario2-u}). The rationale behind this is that when the weighting parameter $\omega_i$ is big for CAV $i$, its optimization boils down to the minimization of the relative distance error. If $\omega_i$ is relatively small, then the control effort minimization term in the cost function becomes more significant so that the relative distance reference is not accomplished exactly.

\begin{figure}
     \centering
     \begin{subfigure}[b]{0.49\textwidth}
         \centering
         \includegraphics[width=\textwidth]{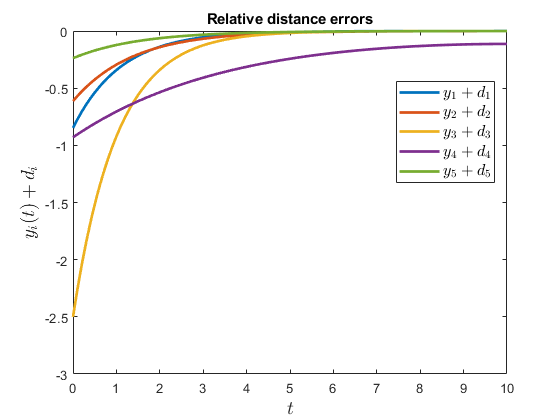}
         \caption{}
         \label{fig:scenario2-y}
     \end{subfigure}%
     \hfill
     \begin{subfigure}[b]{0.49\textwidth}
         \centering
         \includegraphics[width=\textwidth]{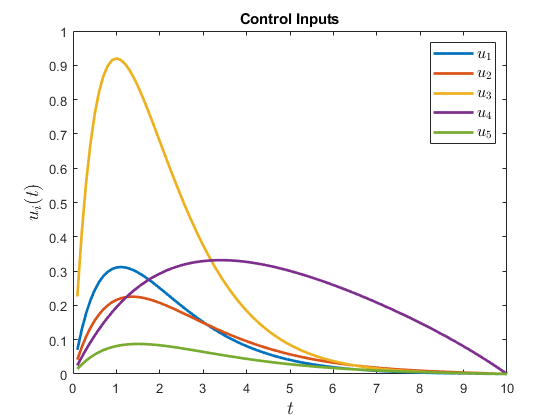}
         \caption{}
         \label{fig:scenario2-u}
     \end{subfigure}%
     \hfill
        \caption{Time histories of (a) relative distance errors and (b) control inputs under the PF topology. The parameter values are set according to Scenario 2 column of Table \ref{tab:PF}.}
        \label{fig:scenario2}
\end{figure}

\textbf{Scenario 3.} The parameter values are given on the left part of Table \ref{tab:TPF}. The time histories of relative distance errors and control inputs for this scenario are shown in Figure \ref{fig:scenario3}. Compared with time history plots in previous scenarios when the underlying topology is PF, here, we see that convergence happens faster under the TPF topology. Starting from CAV 4, each vehicle receives information from two links, one sensor link, and one V2V communication link. This accelerates the convergence of the trajectories and control inputs.

\begin{table}[ht]\caption{Simulation model's parameter values under the TPF topology}
\centering
 \begin{tabular}{|c| c c c c | c c c c |} 
 \hline
 & &   Scenario 3& & &  &  Scenario 4 && \\
 \hline
 $i$& $\Tilde{\omega}_i$& $\omega_i$ &  $d_i$ & $x_i(0)$ & $\Tilde{\omega}_i$ & $\omega_i$ &  $d_i$ & $x_i(0)$ \\ [0.5ex] 
 \hline
 0 &   -         &   -    &   -   & 5.2747 &  -        &  -    &   - & 6.4947 \\ 
 1 &   -         & 0.9157 & -0.1 & 4.8244 &  -         &0.6948 & -0.3 & 5.7887\\ 
 2 &   -         & 0.7922 & -0.3 & 4.7875 &  -         &0.3171 & -0.3 & 3.7157\\
 3 &   0.8491    & 0.9595 & -0.2 & 2.7344 &  0.3816    &0.9502 & -0.2 & 3.2774\\
 4 &   0.9340    & 0.6557 & -0.1 & 1.3925 &  0.7655    &0.0344 & -0.1 & 1.9611\\
 5 &   0.6787    & 0.0357 & -0.3 & 0.4877 &  0.7952    &0      & -0.1 & 0.8559\\ [1ex] 
 \hline
 \end{tabular}\label{tab:TPF}
\end{table}

\begin{figure}
     \centering
     \begin{subfigure}[b]{0.49\textwidth}
         \centering
         \includegraphics[width=\textwidth]{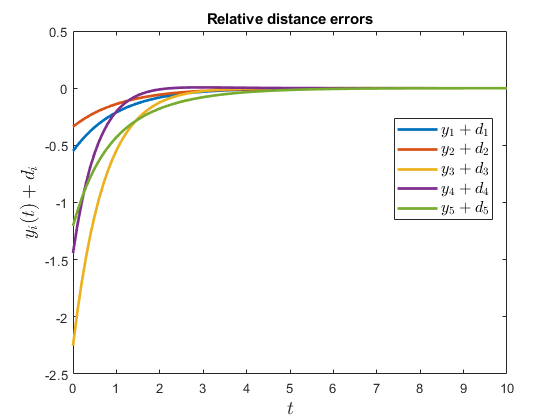}
         \caption{}
         \label{fig:scenario3-y}
     \end{subfigure}%
     \hfill
     \begin{subfigure}[b]{0.49\textwidth}
         \centering
         \includegraphics[width=\textwidth]{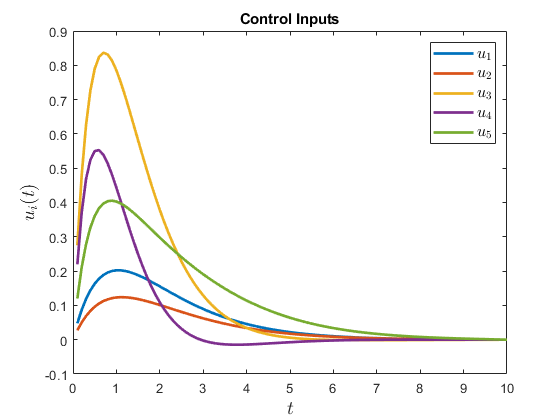}
         \caption{}
         \label{fig:scenario3-u}
     \end{subfigure}%
     \hfill
        \caption{Time histories of (a) relative distance errors and (b) control inputs under the TPF topology. The parameter values are set according to Scenario 3 column of Table \ref{tab:TPF}.}
        \label{fig:scenario3}
\end{figure}

\textbf{Scenario 4.} In another experiment under the TPF topology, we simulate a platoon formation control model with the parameter values given on the right part of Table \ref{tab:TPF}. Here, we assume that the CAV 5's sensor is malfunctioning and this vehicle intends to participate in the platoon and acquire its relative position within the formation only using the V2V communication link (i.e., $\omega_5=0$). From Figure \ref{fig:scenario4}, we observe that all CAVs, including CAV 5, acquire their predetermined inter-vehicle spacing policies within the specified terminal time.

\begin{figure}
     \centering
     \begin{subfigure}[b]{0.49\textwidth}
         \centering
         \includegraphics[width=\textwidth]{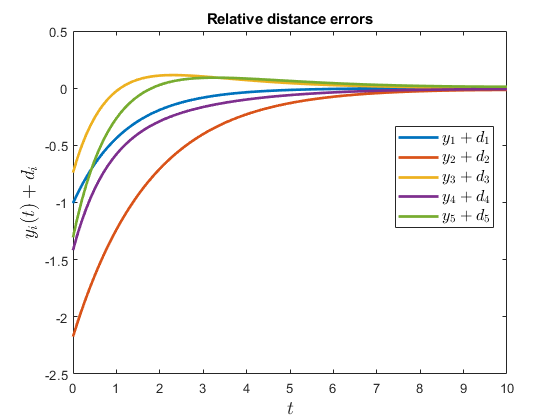}
         \caption{}
         \label{fig:scenario4-y}
     \end{subfigure}%
     \hfill
     \begin{subfigure}[b]{0.49\textwidth}
         \centering
         \includegraphics[width=\textwidth]{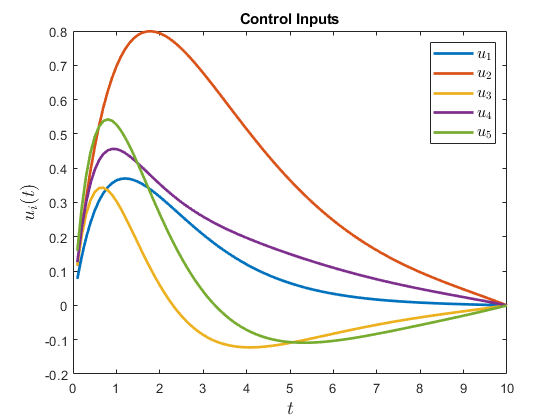}
         \caption{}
         \label{fig:scenario4-u}
     \end{subfigure}%
     \hfill
        \caption{Time histories of (a) relative distance errors and (b) control inputs under the TPF topology. The parameter values are set according to Scenario 4 column of Table \ref{tab:TPF}.}
        \label{fig:scenario4}
\end{figure}

So far, we have investigated the effectiveness of the models and validated the closed-form solutions under the PF and TPF topology within the simulation Scenarios 1-4. Here, we expand these simulation studies to general topology within the simulation Scenarios 5-6 with the model's parameter values given in Table \ref{tab:general}.    

\textbf{Scenarios 5-6.} All predecessor following topology (APF) and leader following topology (LF) are considered in these scenarios. The time histories of relative distance errors for these topology are shown in Figure \ref{fig:scenario5-6}\subref{fig:scenario5} and Figure \ref{fig:scenario5-6}\subref{fig:scenario6}, respectively. As it is seen from the figures, the CAVs under the APF topology converge faster while they can not accomplish the predefined inter-vehicle spacing policies exactly under the LF topology in the time horizon $t_f=10$. The rationale behind is missing the sensor links in the LF topology. 

\begin{table}[ht]\caption{Simulation model's parameter values under general topology}
\centering
 \begin{tabular}{|c|c c c c | c c c c |} 
 \hline
 & &   Scenario 5 & &  & & Scenario 6& &   \\
 \hline
 $i$ &    $\mathcal{N}_i$& $\omega_{ij}$   &$d_i$& $x_i(0)$& $\mathcal{N}_i$& $\omega_{ij}$ &  $d_i$ & $x_i(0)$ \\ [0.5ex] 
 \hline
 0  &     -      &   -                        & -          &   5.5166  &       -     &    -     &    -    &   4.6854      \\ 
 1  & \{0\}      &  0.6324         &  -0.1      &   4.7511  &    \{0\}    &  0.2769  & -0.2    &   3.9760     \\ 
 2  & \{1\}      &  1.0975  &  -0.2      &   2.1937  &    \{1\}    &  0.0462  &  -0.2   &   3.2316      \\
 3  & \{1,2\}    &   \{0.7547,1.2785\}                   &  -0.3      &   1.9078  &    \{1\}    &  0.0971  &  -0.1   &   2.4488      \\
 4  & \{1,2,3\}  &  \{0.2760,0.5469,2.9134\}  &  -0.2      &   1.5855  &    \{1\}    &  0.0935  &  -0.3   &   2.2279       \\
 5  & \{1,2,3,4\}&  \{0.9649,0.8147,0.1576,1.9706\}  &  -0.1      &   0.1722  &    \{1\}    &  0.1948  &  -0.2   &   0.9344      \\ [1ex] 
 \hline
 \end{tabular}\label{tab:general}
\end{table}

\begin{figure}
     \centering
     \begin{subfigure}[b]{0.49\textwidth}
         \centering
         \includegraphics[width=\textwidth]{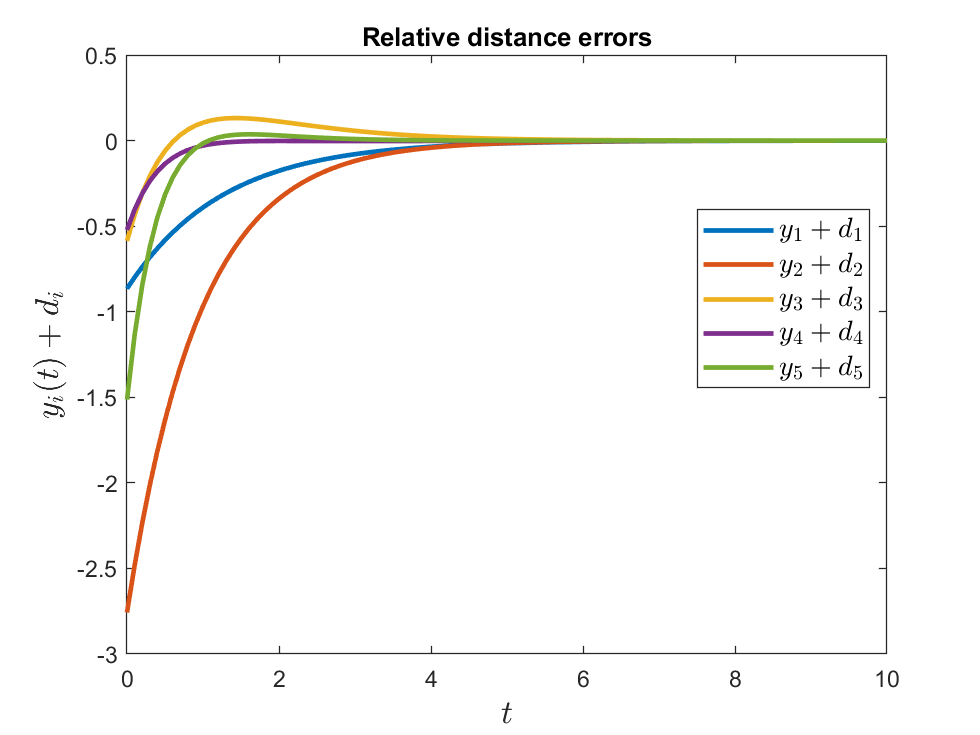}
         \caption{}
         \label{fig:scenario5}
     \end{subfigure}%
     \hfill
     \begin{subfigure}[b]{0.49\textwidth}
         \centering
         \includegraphics[width=\textwidth]{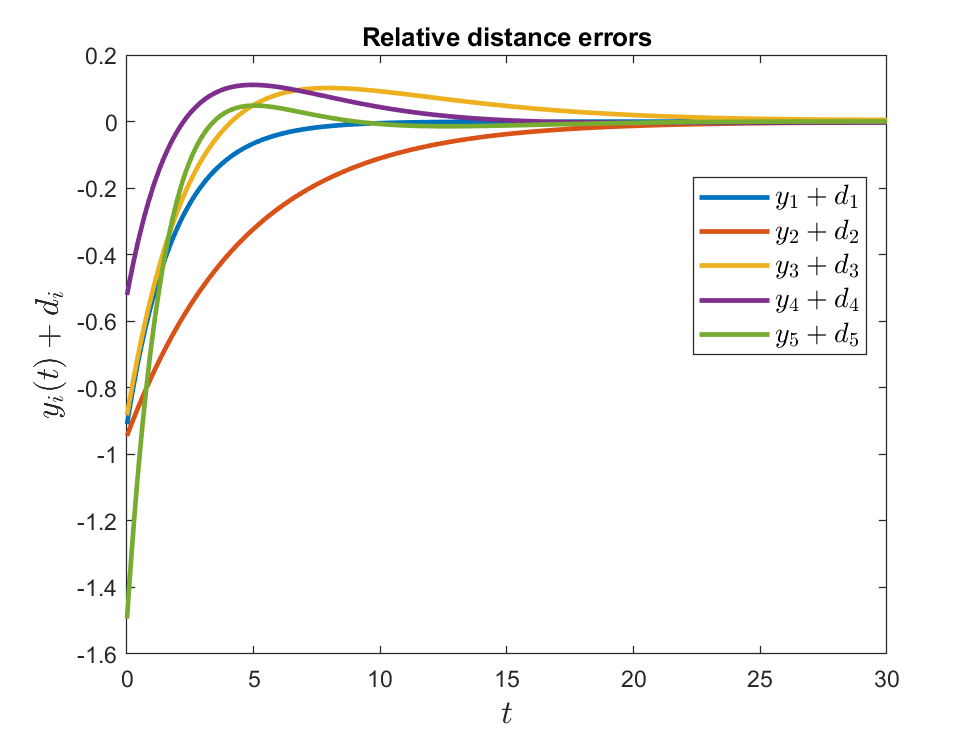}
         \caption{}
         \label{fig:scenario6}
     \end{subfigure}%
     \hfill
        \caption{Time histories of relative distance errors (a) all predecessor following topology (APF) and (b) leader following topology (LF). The parameter values are set according to Table \ref{tab:general}.}
        \label{fig:scenario5-6}
\end{figure}

\section{Conclusions}\label{section:conclusions}
In this paper, the vehicle platoon formation control has been formulated with noncooperative differential games. The existence of a  Nash equilibrium controller, which vehicles use as their self-enforcing control strategy to acquire the platoon formation, for CAVs is proved. Additionally, the CAVs individual relative distance trajectories are obtained in the closed-form under the PF and TPF topology. The simulation studies showed a platoon with 5 vehicles acquired their predetermined inter-vehicle spacing policies utilizing their Nash equilibria. In this paper, the CAV platoon formation control has been formulated with noncooperative differential games. The existence of a Nash equilibrium controller, which CAVs use as their self-enforcing control strategy to acquire the platoon formation, is proved. Additionally, the CAVs individual relative distance trajectories are obtained in the closed-form under the PF and TPF topology and can be calculated numerically under general topology. The simulation studies showed a platoon with 5 vehicles acquired their predetermined inter-vehicle spacing policies utilizing their Nash equilibria. The proposed approach is powerful and promising which enables some future research directions. First,  state constraints can be incorporated into the problem. The technique developed for PA and PFA are analytical approaches and general topology is semi-analytical. Therefore, the question arises of whether numerical solutions to constrained optimization problems can be obtained. The candidates for such numerical techniques are the Shooting method, Dynamic programming, and the Hamilton-Jacobi-Bellman technique. Another extension is to consider general LTI systems or double integrators instead of single integrator dynamic constraints studied here. We might also suppose arbitrary final formation in the optimization problem on top of the unspecified terminal condition utilized in this study. Finally, we can include double states which represent the two-dimensional maneuvers as well.

\subsection*{Acknowledgements} 
%An Acknowledgements section is started with \verb"\ack" or \verb"\acks" for \textit{Acknowledgement} or \textit{Acknowledgements}, respectively. It must be placed just before the References.
This work was supported by SGS, VSB - Technical University of Ostrava, Czech Republic, under the grant No. SP2021/24 “Parallel processing of Big Data VIII”.
This work was also supported in part by the Science and Research Council of
Turkey (T\"{U}B\.{I}TAK) under project EEEAG-121E162.

\end{document}